\documentclass[conference]{IEEEtran}
\IEEEoverridecommandlockouts
\usepackage{cite}
\usepackage{amsmath,amssymb,amsfonts}
\usepackage{graphicx}
\usepackage{textcomp}
\usepackage{xcolor}
\def\BibTeX{{\rm B\kern-.05em{\sc i\kern-.025em b}\kern-.08em
    T\kern-.1667em\lower.7ex\hbox{E}\kern-.125emX}}
\usepackage{url}
\usepackage{algorithm}
\usepackage{algpseudocode}
\usepackage[acronym]{glossaries}
\newacronym{deepc}{DeePC}{data-enabled predictive control}
\newacronym{spc}{SPC}{subspace predictive control}
\newacronym{lti}{LTI}{linear time-invariant}
\newacronym{mpc}{MPC}{model predictve control}
\newacronym{mimo}{MIMO}{multiple-input, multiple-output}
\newacronym{siso}{SISO}{single-input, single-output}
\newacronym{pe}{PE}{persistently exciting}
\newacronym{ocp}{OCP}{optimal control problem}
\newacronym{dpc}{DPC}{data-driven predictive control}
\newacronym{pwa}{PWA}{piecewise affine}
\newacronym{licq}{LICQ}{linear independence constraint qualification}

\newtheorem{defn}{Definition}
\newtheorem{lem}{Lemma}
\newtheorem{assum}{Assumption}
\newtheorem{prop}{Proposition}
\begin{document}

\title{Data-driven predictive control of nonlinear systems using weighted regularization\\
\thanks{This work is part of Hollandse Kust Noord wind farm innovation program where CrossWind C.V., Shell, Eneco, Grow and Siemens Gamesa are teaming up; funding for the PhD’s and PostDocs was provided by CrossWind C.V. and Siemens Gamesa.}}

\author{\IEEEauthorblockN{Fritz A. Engeln\IEEEauthorrefmark{1}, Sebastian Zieglmeier\IEEEauthorrefmark{2}, Marta Zagórowska\IEEEauthorrefmark{1} and Jan-Willem van Wingerden\IEEEauthorrefmark{1}}
\IEEEauthorblockA{\IEEEauthorrefmark{1}Delft Center for Systems and Control, Delft University of Technology, Delft, The Netherlands\\ (email: \{f.a.engeln,m.a.zagorowska,j.w.vanwingerden\}@tudelft.nl)}
\IEEEauthorblockA{\IEEEauthorrefmark{2}Department of Technology Systems, University of Oslo, Kjeller, Norway (email: sebastiz@uio.no)}}

\maketitle

\begin{abstract}
Data-driven control methods, like Data-enabled Predictive Control (DeePC), are often formulated for linear systems, where the principle of superposition allows global system behavior to be inferred from locally collected data through \emph{Willems' fundamental lemma}. This principle does not hold for nonlinear systems, whose dynamics may vary across operating regions. 
We propose a data-driven predictive control framework for nonlinear systems that incorporates data column preferences according to their proximity to the current operating point through a weighted norm regularization, thereby localizing the predictor without discarding any data. 
We show how the proposed weighting scheme induces operating point-dependent data prioritization and ensures a well-posed optimization problem.
A numerical study on a nonlinear two-tank system demonstrates that the proposed method matches or outperforms hard data-selection schemes while retaining the full data matrix and its rank, thereby guaranteeing feasibility.

\end{abstract}


\begingroup
\renewcommand\thefootnote{}
\footnotetext{\hspace*{-1.1em}The code to reproduce the results of this paper is available at: \url{https://github.com/fengeln/weightDPC}}
\endgroup
\section{Introduction}
Classical predictive control methods such as model predictive control \cite{rawlings2017model} rely on an accurate model of the system, obtained either from first-principles modeling or system identification \cite{ljung1999system}. More recently, increasing attention has been directed toward deriving control policies directly from data, thereby avoiding the need for an explicit parametric model of the system. Instead, \emph{Willems' fundamental lemma} \cite{willems2005note} is used to characterize the behavior of an unknown \gls{lti} system by the linear subspace in which all input-output trajectories reside. Leveraging the \emph{fundamental lemma}, \gls{dpc} schemes can be formulated, such as \gls{deepc}~\cite{coulson2019deepc}. 

Since the \textit{fundamental lemma} presumes a noise-free \gls{lti} system, \gls{deepc} inherits these limitations. Introducing regularization relaxes the strict consistency requirement and lets \gls{deepc} tolerate measurement noise and mild nonlinearities, while also providing closed-loop stability and robustness guarantees~\cite{Berberich_2020}. When the system exhibits nonlinear behavior that cannot be well approximated by a single global linear model, controller performance deteriorates because trajectories recorded in one operating region carry little information about the local behavior in another, causing both prediction accuracy and closed-loop performance to degrade~\cite{zieglmeier2025gain}.

A recent and effective remedy is to localize the data online. Select-\gls{dpc}, introduced in \cite{naf2026choosewiselydatadrivenpredictive}, retains at each time step only the columns of the data matrix that are closest to the current operating point according to either a norm- or a manifold-based metric. This approach implicitly linearizes the dynamics in the trajectory space while preserving the convexity of the \gls{ocp} and allowing the standard \gls{deepc} formulation to be reused. A closely related approach was independently proposed in~\cite{beerwerth2025less}.
Both methods show that localizing data online is an effective way to handle nonlinear dynamics, yet rely on a hard, discrete selection in which each column is either kept or discarded. This entails two drawbacks. First, individually informative columns may be discarded~\cite{li2025datamodel}, while the retained columns can become highly similar, resulting in Hankel matrices that span a poorly conditioned subspace. Second, because the retained data changes from one time step to the next, standard theoretical guarantees of \gls{deepc}, such as recursive feasibility, do not directly carry over.

We address the limitations of hard data selection by retaining all data columns and assigning each a weight based on its distance to the current operating point. These weights are incorporated through a positive-definite weighting matrix in the regularization of the column selector, such that nearby trajectories contribute more strongly to the prediction, while distant trajectories are penalized. We position our method, weight-\gls{dpc}, between standard \gls{deepc} and select-\gls{dpc}. \gls{deepc} weights every column equally, while select-\gls{dpc} keeps a hard subset and discards the rest. In contrast to the weighted regularization approach proposed in \cite{li2026rdsdeepcrobustdataselection}, we define the weighting matrix explicitly and provide a theoretical analysis of the resulting algorithm. Unlike our distance-based weighting, exponentially decaying time-dependent weights are proposed in \cite{10885954} as a forgetting factor for adaptive control.

The contributions of this paper are threefold: (i) We propose weight-\gls{dpc}, where the usage of each column of the data matrix is penalized with a weight that reflects the distance of the data to the current operating point of the system.
(ii) We establish well-posedness (recursive feasibility and uniqueness) for $\ell_2$- and weighted $\ell_2$-norm regularization, showing that a non-uniform weighting enables prioritization of locally relevant data.
(iii) We validate the theoretical results on a nonlinear two-tank system by empirically analyzing the column selector $g$, highlighting the differences between \gls{deepc}, select-\gls{dpc}, and weight-\gls{dpc}.

The remainder of the paper is organized as follows. Section~\ref{sec:Preliminaries} introduces the model class and the preliminaries on \gls{dpc}. Section~\ref{sec:NLDPC} develops the weight-\gls{dpc} formulation and Section~\ref{sec:theoretical} its theoretical framework. Section~\ref{sec:numericalEx} presents the numerical study, and Section~\ref{sec:Conclusion} concludes.

\textit{Notation:} The identity matrix is denoted by \(I\) and the Kronecker product by $\otimes$. For a symmetric matrix \(X\), the notation \(X \succ 0\) indicates that \(X\) is positive definite. Any matrix \(X \succ 0\) induces the weighted Euclidean norm \(\|v\|_X=\sqrt{v^\top Xv}\) for all vectors $v$ of appropriate dimension. The standard Euclidean norm (\(\ell_2\)-norm) is denoted by \(\|v\|=\sqrt{v^\top v}\) which corresponds to the special case \(X = I\).
\section{Preliminaries}
\label{sec:Preliminaries}
We consider the time-invariant nonlinear state-space model
\begin{equation}
\begin{aligned}
\label{eqn:nonlinearSS}
    x_{k+1} &= \phi(x_k,u_k)\\
    y_k &= \psi(x_k, u_k),
\end{aligned}
\end{equation}
where \(x_k \in \mathbb{R}^{n}\), \(u_k\in \mathbb{R}^{n_u}\), and \(y_k\in \mathbb{R}^{n_y}\) are the state, input, and output vector at time index \(k \in \mathbb{Z}_{>0}\), respectively. The state transition function \(\phi: \mathbb{R}^{n} \times \mathbb{R}^{n_u} \rightarrow \mathbb{R}^{n}\) and the output function \(\psi: \mathbb{R}^{n} \times \mathbb{R}^{n_u} \rightarrow \mathbb{R}^{n_y}\) are in general nonlinear functions.\\
%
%
\subsection{The fundamental lemma}
We denote a stacked vector of \(T\) consecutive input samples starting at time \(\hat{i}\) by \(\underline{u}_{\hat{i},T}\in \mathbb{R}^{{n_u}T}\) and a block-Hankel matrix of depth \(L\) constructed from the samples in \(\underline{u}_{\hat{i},T}\) by
\begin{align}
    U_L &= \begin{bmatrix}u_{\hat{i}}&\cdots&u_{\hat{i}+T-L}\\\vdots&\ddots&\vdots\\u_{\hat{i}+L-1}&\cdots&u_{\hat{i}+T-1}\end{bmatrix}\in \mathbb{R}^{n_uL \times n_g}. 
\end{align}
Similarly, the matrix \(Y_L\in \mathbb{R}^{n_yL \times n_g}\) is constructed from \(\underline{y}_{\hat{i},T}\) with $n_g =T-L+1$.
\begin{defn}
    A sequence \(\underline{u}_{\hat{i},T}\) is said to be \gls{pe} of order \(L\) if \(\mathrm{rank}(U_L)=n_uL\).
\end{defn}
\begin{lem}[Willems' fundamental lemma, \cite{willems2005note}]
\label{lem:Willems}
    Consider recorded input-output trajectories \(\underline{u}_{\hat{i},T}\) and \(\underline{y}_{\hat{i},T}\) of an \(n^{th}\)-order controllable \gls{lti} system, where the input data \(\underline{u}_{\hat{i},T}\) is PE of order \(L+n\) and suppose \(U_L\) and \(Y_L\) to be constructed from that data. Then, the trajectories \(\underline{u}_{i,L}\) and \(\underline{y}_{i,L}\), starting at any time \(i \in\mathbb{Z}_{>0}\) belong to the system if and only if there exists a \(g\in \mathbb{R}^{n_g}\) such that
    \begin{align}
        \label{eqn:implicitPredictor}
        \begin{bmatrix}U_L\\Y_L\end{bmatrix}g&=\begin{bmatrix}\underline{u}_{i,L}\\ \underline{y}_{i,L}\end{bmatrix}.
    \end{align}
\end{lem}
Lemma~\ref{lem:Willems} is a key result in behavioral systems theory and implies that the data matrix has \(n_uL+n\) linearly independent columns that span the trajectory space of the system. This result has implications for system identification and control synthesis, as it allows system properties to be inferred directly from data, bypassing the need for explicit model identification.
However, Lemma~\ref{lem:Willems} is limited to noise-free controllable \gls{lti} systems. In contrast to nonlinear systems, linear systems follow the principle of superposition and any input-output trajectory can be constructed from linear combinations of columns of the data matrix.
\subsection{Data-driven predictive control}
The \textit{fundamental lemma} can be leveraged to design a \gls{dpc} scheme \cite{7244723,coulson2019deepc}. Given past and future time horizons \(p\) and \(f\), with \(L=p+f\), we partition the data in~\eqref{eqn:implicitPredictor}, such that
\begin{align}
    U_L &= \begin{bmatrix} U_p \\ U_f \end{bmatrix} \quad\text{and}\quad  Y_L = \begin{bmatrix} Y_p \\ Y_f \end{bmatrix}. 
\end{align}
Furthermore, we have the past input data \(\underline{u}_{i,p}\), past output data \(\underline{y}_{i,p}\), a reference trajectory \(\underline{r}_{i+p,f}\), an input constraint set \(\mathcal{U} \subseteq \mathbb{R}^{n_u}\), an output constraint set \(\mathcal{Y} \subseteq \mathbb{R}^{n_y}\) and formulate the following \gls{ocp}:
\begin{subequations}
\label{eqn:dpc}
    \begin{align}
        &\min_{g,\sigma,\underline{u}_{i+p,f},\underline{y}_{i+p,f}}\!\!\!\!J(\underline{u}_{i+p,f},\underline{y}_{i+p,f})\!  +\! \lambda_g\|g\|^2 \!+\!\lambda_\sigma\|\sigma\|^2\\
        \label{eqn:DeePCConstraint}
        \mathrm{s.t.} &\begin{bmatrix}U_p\\U_f\\Y_p\\Y_f\end{bmatrix}g=\begin{bmatrix}\underline{u}_{i,p}\\\underline{u}_{i+p,f}\\\underline{y}_{i,p}+\sigma\\\underline{y}_{i+p,f}\end{bmatrix}\\
        \label{eqn:IOconstraints}
        &(u_k, y_k) \in \mathcal{U} \times \mathcal{Y}, \; \forall k \in \{i\!+\!p, ...,i\!+\!p\!+\!f\!-\!1 \},
    \end{align}
\end{subequations}
where the scalars $\lambda_\sigma$ and $\lambda_g$ are weights on the penalization of the slack vector $\sigma \in \mathbb{R}^{n_yp}$ and the column selector $g$, respectively. The objective function is given by
\begin{align}
\label{eqn:Objective}
    J(\underline{u}_{i+p,f},\underline{y}_{i+p,f}) &= \sum_{k=i+p}^{i+p+f-1}\underbrace{\|y_k-r_k\|_Q^2 + \|u_k\|_R^2}_{\ell(u_k,y_k)},
\end{align}
with the performance matrices \(Q,R \succ 0 \) and the stage cost $\ell(u_k,y_k)$.\\
A \gls{dpc} scheme is obtained by solving the \gls{ocp} in~\eqref{eqn:dpc} in a receding horizon fashion and applying the first input of the optimal input trajectory to the system. At each iteration, the past input and output trajectories are updated with the past \(p\) input-output measurements, which are assumed to be known to implicitly initialize the system. For deterministic systems, the minimum past horizon length corresponds to the lag of the system \cite{Markovsky01122008}. Additionally, the input and output constraints in~\eqref{eqn:IOconstraints} can be chosen to enforce the system to operate within physical and safety limits.\\
Considering solely the data equality constraint in~\eqref{eqn:DeePCConstraint} while neglecting the relaxation $\sigma$, the data matrix on the left-hand side has full row rank if the output data is corrupted by noise and the columns span the linear space of all trajectories of length \(L\). In other words, the solutions to the \gls{ocp} are not restricted to the systems dynamics and the constraint~\eqref{eqn:DeePCConstraint} is meaningless. To restore the meaning of the constraint, the column selector \(g\) is regularized. Both, the $\ell_1$-norm and $\ell_2$-norm are commonly used for regularization, while recent results have shown that the $\ell_2$-norm regularization is favorable \cite{11020753}. The next section introduces a weighted regularization that prioritizes locally relevant data for nonlinear systems with operating point-dependent behavior.
\section{Data-driven predictive control for nonlinear systems}
\label{sec:NLDPC}
For nonlinear systems such as~\eqref{eqn:nonlinearSS}, the dynamics may vary across the state space. We therefore seek a trajectory-based measure that quantifies the similarity between the current operating point and the operating points represented by the columns of the data matrix in~\eqref{eqn:implicitPredictor}.
\subsection{Distance between trajectories}
In the behavioral framework, the system state is not explicitly available. Hence, quantifying the distance between operating points must be based on input-output trajectories. At time $i+p$ in~\eqref{eqn:dpc}, the past $p$ input-output measurements are available. For a controllable \gls{lti} system \cite{ljung1999system}, the state can be expressed as
\begin{align}
    x_{i+p} &= \mathcal{K}\begin{bmatrix}\underline{u}_{i,p}\\\underline{y}_{i,p}\end{bmatrix} = \mathcal{K}\tau_{i,p},    
\end{align}
where $\mathcal{K}$ denotes the extended controllability matrix \cite{ENGELN2026100384}, and $\tau_{i,p}$ is a vector with the past $p$ input-output samples.\\
Let $\tau_{v,p}$ denote a vector containing the first $p$ input and output samples of the column of the data matrix in~\eqref{eqn:DeePCConstraint} starting at time index \(v\), such that
\begin{align}
    \tau_{v,p} &= \begin{bmatrix}\underline{u}_{v,p}\\ \underline{y}_{v,p}\end{bmatrix}, \quad \forall v \in \{\hat{i}, ...,\hat{i}+n_g-1\}.
\end{align}
The distance between the current state of an \gls{lti} system and the state associated with the $v^{th}$-column of the data matrix after $p$ time steps can be written as
\begin{align}
\label{eqn:distanceLTI}
    \|x_{v+p} - x_{i+p}\| &=  \|\mathcal{K}(\tau_{v,p} - \tau_{i,p})\| .
\end{align}
Although~\eqref{eqn:distanceLTI} is only valid for linear systems and $\mathcal{K}$ is generally unknown, it motivates measuring similarity directly in trajectory space. We therefore define
\begin{align}
\label{eqn:distancew}
    w_v &= \|\tau_{v,p} - \tau_{i,p}\|, \quad \forall v \in \{\hat{i}, ...,\hat{i}+n_g-1\},
\end{align}
where small values of $w_v$ indicate operating conditions of the $v^{th}$ column similar to the current one.\\
Unlike~\eqref{eqn:distancew}, the Euclidean distance metric in \cite{naf2026choosewiselydatadrivenpredictive} uses complete trajectories of length $L$ and therefore requires an approximation of the unknown future trajectory from the previous \gls{ocp} solution.
\subsection{Data preferences}
Assume that input-output data of the nonlinear system in~\eqref{eqn:nonlinearSS} have been collected from multiple operating regions. As a result, the data matrix captures a broad range of local system behaviors. We assign a weight to each column of the data matrix based on its distance in \eqref{eqn:distancew} to the current operating point. Let the vector
\begin{align}
    \underline{w}_{\hat{i},n_g} &= \begin{bmatrix}w_{\hat{i}}& \cdots & w_{\hat{i}+n_g-1}\end{bmatrix}
\end{align}
contain the distances of each column of the data matrix to the current operating point as defined in~\eqref{eqn:distancew}. We define the normalized distances as
\begin{align}
\label{eqn:weights}
    \tilde{w}_v &= \frac{w_v}{\|\underline{w}_{\hat{i},n_g}\|} , \quad \forall v \in \{\hat{i}, ...,\hat{i}+n_g-1\},
\end{align}
such that 
\begin{align}
    \underline{\tilde{w}}_{\hat{i},n_g} &= \begin{bmatrix}\tilde{w}_{\hat{i}}& \cdots & \tilde{w}_{\hat{i}+n_g-1}\end{bmatrix}
\end{align}
and \(\|\underline{\tilde{w}}_{\hat{i},n_g}\|=1\). Next, we define the diagonal matrix
\begin{align}
\label{eqn:weight_matrix}
    \Lambda &= \lambda_w\begin{bmatrix}\tilde{w}_{\hat{i}}&0&\cdots&0\\0&\tilde{w}_{\hat{i}+1}&\ddots&0 \\ \vdots & \ddots & \ddots &\vdots\\0&0& \cdots &\tilde{w}_{\hat{i}+n_g-1}  \end{bmatrix} + \lambda_g I
\end{align}
with the scalar weight parameters $\lambda_w\geq0$ and $\lambda_g>0$. Since $\Lambda$ is a diagonal matrix, its eigenvalues are given by $\lambda_w\tilde{w}_v + \lambda_g$. Therefore, the matrix $\Lambda$ is positive definite and induces a \(2\)-norm. With this, we introduce the weight-\gls{dpc} \gls{ocp}:
\begin{subequations}
\label{eqn:OCPNLDPC}
    \begin{align}
        &\min_{g,\sigma,\underline{u}_{i+p,f},\underline{y}_{i+p,f}}\!\!J(\underline{u}_{i+p,f},\underline{y}_{i+p,f})\!+\! \|g\|_{\Lambda}^2 \!+\! \lambda_\sigma\|\sigma\|^2
        \label{eq:cost_func}
        \\
        \mathrm{s.t.} &\begin{bmatrix}U_p\\U_f\\Y_p\\Y_f\end{bmatrix}g=\begin{bmatrix}\underline{u}_{i,p}\\\underline{u}_{i+p,f}\\\underline{y}_{i,p}+\sigma\\\underline{y}_{i+p,f}\end{bmatrix}
        \label{eq:data_constraints}
        \\
        &(u_k, y_k) \in \mathcal{U} \times \mathcal{Y}, \forall k \in \{i\!+\!p, ...,i\!+\!p\!+\!f\!-\!1 \},
        \label{eq:constraintset}
    \end{align}
\end{subequations}
where \(\Lambda\) is weighting the entries of \(g\) according to their distance of the current operation point. Columns corresponding to nearby operating conditions are penalized less and are therefore favored in the construction of the prediction. The \gls{ocp} is solved in a receding horizon fashion and the matrix \(\Lambda\) is calculated at every time step. The weight-\gls{dpc} algorithm is summarized in Algorithm~\ref{alg:nonlinear_dpc}.
\begin{algorithm}
\caption{Weight-\gls{dpc}}
\label{alg:nonlinear_dpc}
\begin{algorithmic}[1]
\State \textbf{Given:} \(\underline{u}_{i,p}\), and \(\underline{y}_{i,p}\)
\While{control task not finished}
    \State Calculate the matrix \(\Lambda\) with~\eqref{eqn:weight_matrix}
    \State Solve the \gls{ocp} in~\eqref{eqn:OCPNLDPC} for \(\underline{u}_{i+p,f}^\ast\) and \(\underline{y}_{i+p,f}^\ast\)
    \State Apply the first control input $u_{i+p} = u_{i+p}^\ast$
    \State Set $i \gets i + 1$
    \State Update \(\underline{u}_{i,p}\) and \(\underline{y}_{i,p}\) with the most recent values
\EndWhile
\end{algorithmic}
\end{algorithm}
\section{Optimal column selector}
\label{sec:theoretical}
To analyze the effect of the weighted regularization term in~\eqref{eq:cost_func} on the utilization of the data matrix columns, we derive an explicit solution to the weight-\gls{dpc} optimization problem in~\eqref{eqn:OCPNLDPC}. The resulting column selector is then compared with the standard \gls{deepc} formulation employing $\ell_2$-norm regularization as given in~\eqref{eqn:dpc}. Throughout the analysis, the following assumptions are imposed:
\begin{assum}
\label{ass:Constraints}
    The input and output constraint sets are polyhedral
    \begin{equation}
    \begin{aligned}
        \mathcal{U} &= \{u\in\mathbb{R}^{n_u}:A_u u \leq b_u\}\quad\\
        \mathcal{Y} &= \{y\in\mathbb{R}^{n_y}:A_y y\leq b_y\}
    \end{aligned}
    \end{equation}
    for suitable matrices $A_u$, $A_y$ and vectors $b_u$, $b_y$.
\end{assum}
\begin{assum}
\label{ass:inputdata}
    The input sequence used for data collection is \gls{pe} of order $L$, such that the Hankel matrix $U_L$ has full row rank.
\end{assum}
\begin{assum}
\label{ass:rankDataMatrix}
    The output Hankel matrix $Y_L$ has full row rank, and its rows are linearly independent of the rows of $U_L$.
\end{assum}
Under Assumption~\ref{ass:Constraints}, the optimization problem in~\eqref{eqn:OCPNLDPC} can be reformulated solely in terms of the column selector vector $g$
\begin{subequations}
\label{eqn:OCPing}
\begin{align}
\label{eqn:OCPingObjec}
    \min_g &\; J(U_fg,Y_fg) + \|g\|_\Lambda^2 + \lambda_\sigma\|Y_pg-\underline{y}_{i,p}\|^2\\
    \mathrm{s.t\;}& U_pg = \underline{u}_{i,p}\\
    \label{eqn:constraintsOCPing}
    &Ag \leq b,
\end{align}
\end{subequations}
where $A$ and $b$ represent the polyhedral input-output constraints induced by Assumption~\ref{ass:Constraints}. Furthermore, Assumptions~\ref{ass:inputdata} and~\ref{ass:rankDataMatrix} imply that the combined data matrix $\begin{bmatrix}U_L^\top&Y_L^\top\end{bmatrix}^\top$ has full row rank.\\
Before deriving an explicit expression for the solution of~\eqref{eqn:OCPing}, we establish that the optimization problem is well posed. In particular, we show that the problem is recursively feasible and admits a unique minimizer.
\subsection{Well-posedness}
\begin{defn}[Recursive feasibility, \cite{rawlings2017model}]
    An \gls{ocp} is said to be recursively feasible if feasibility at time $k$ implies feasibility at time $k+1$ for all $k\in\mathbb{N}$.
\end{defn}
\begin{lem}
\label{lem:wellposedness}
Let the input and output constraint sets be nonempty, i.e., $\mathcal{U}\neq\emptyset$ and $\mathcal{Y}\neq\emptyset$. Then the \gls{ocp} in~\eqref{eqn:OCPing} is recursively feasible. Furthermore, if $Q,R\succ0$ and the weight matrix $\Lambda$ is given by~\eqref{eqn:weight_matrix}, the optimization problem admits a unique minimizer.
\end{lem}
\begin{IEEEproof}
    Assume that at time $i+p$ there exists a feasible solution with optimal input and output trajectories
    \begin{equation}
    \begin{aligned}
        \underline{u}^\ast_{i+p,f} &= \begin{bmatrix}u^\ast_{i+p}\\ \vdots \\ u^\ast_{i+p+f-1}\end{bmatrix}  \quad
        \underline{y}^\ast_{i+p,f} = \begin{bmatrix}y^\ast_{i+p}\\\vdots\\y^\ast_{i+p+f-1}\end{bmatrix}, 
    \end{aligned}
    \end{equation}
    which satisfy the constraints of~\eqref{eqn:OCPing} for given $\underline{u}_{i,p}$ and $\underline{y}_{i,p}$. After applying the first optimal input $u^\ast_{i+p}$ and measuring the output $y_\mathrm{meas}$ the past data at time $i+p+1$ is updated to
    \begin{align}
        u_{\mathrm{ini}} &= \begin{bmatrix}u_{i+1}\\ \vdots \\u_{i+p-1}\\u^\ast_{i+p}\end{bmatrix} \quad y_{\mathrm{ini}} = \begin{bmatrix} y_{i+1}\\ \vdots \\y_{i+p-1}\\y_{\mathrm{meas}}\end{bmatrix}. 
    \end{align}
    Consider the shifted candidate trajectories
    \begin{equation}
    \begin{aligned}
        \underline{u}^\ast_{i+p+1,f} &= \begin{bmatrix}{u^\ast}^\top_{i+p+1}&\cdots&{u^\ast}^\top_{i+p+f-1}&u^\top_{\mathrm{app}}\end{bmatrix}^\top  \\
        \underline{y}^\ast_{i+p+1,f} &= \begin{bmatrix}{y^\ast}^\top_{i+p+1}&\cdots&{y^\ast}^\top_{i+p+f-1}&y^\top_{\mathrm{app}}\end{bmatrix}^\top, 
    \end{aligned}
    \end{equation}
    for any $u_{\mathrm{app}}\in\mathcal{U}$ and $y_{\mathrm{app}}\in\mathcal{Y}$. Both $u_\mathrm{app}$ and $y_{\mathrm{app}}$ exist since the constraint sets are nonempty by Assumption~\ref{ass:Constraints}. The constraints from the data-driven model at time $i+p+1$ become
    \begin{align}
    \label{eqn:appendedModel}
        \begin{bmatrix}U_p\\U_f\\Y_p\\Y_f\end{bmatrix}g &= \begin{bmatrix}u_{\mathrm{ini}}\\ \underline{u}^\ast_{i+p+1,f}\\y_{\mathrm{ini}}+\sigma\\\underline{y}^\ast_{i+p+1,f}\end{bmatrix}.
    \end{align}
    By Assumptions~\ref{ass:inputdata} and~\ref{ass:rankDataMatrix}, the matrix on the left-hand side of~\eqref{eqn:appendedModel} has full row rank. Hence, its columns span the entire space $\mathbb{R}^{L(n_u+n_y)}$, implying the existence of a vector $g$ satisfying~\eqref{eqn:appendedModel}. Therefore, a feasible candidate solution exists at time $i+p+1$, establishing recursive feasibility.\\
    Under Assumption~\ref{ass:Constraints}, all constraints in~\eqref{eqn:OCPing} are affine in $g$, and thus define a convex feasible set. Furthermore, $Q,R\succ0$ and $\Lambda\succ0$ imply that the objective function in \eqref{eqn:OCPingObjec} is strictly convex in $g$. Consequently, the optimization problem admits a unique minimizer (cf. \cite{bertsekas1999nonlinear}, Prop.~1.1.2).
\end{IEEEproof}
Lemma~\ref{lem:wellposedness} highlights a key advantage of weight-\gls{dpc}. Since weight-\gls{dpc} retains all columns of the data matrix and only modifies their contribution through the weighting matrix $\Lambda$, the full rank of the data matrix is preserved. In contrast, select-\gls{dpc} removes columns from the data matrix, which may reduce the rank required for feasibility and the informativity of the remaining data \cite{li2025datamodel}. Consequently, recursive feasibility and the existence of a solution are not guaranteed in general for select-\gls{dpc} unless additional mechanisms are employed to preserve the rank of the data matrix.
\subsection{Explicit solution of the column selector}
To characterize the effect of the weighted regularization on the utilization of the data matrix columns, we derive an explicit expression for the unique minimizer $g^\ast$. A related analysis for \gls{deepc} with $\ell_1$-norm regularization is presented in \cite{11312489}. We extend this type of analysis to the proposed weighted $\ell_2$-regularization.\\
Let the active inequality constraints of~\eqref{eqn:constraintsOCPing} be represented by the matrix $\tilde{A}$ and vector $\tilde{b}$ such that $\tilde{A}g^\ast=\tilde{b}$. Combining the equality constraint $U_pg=\underline{u}_{i,p}$ with the active inequality constraints, define $\hat{A} = \begin{bmatrix}U_p^\top&\tilde{A}^\top\end{bmatrix}^\top\in\mathbb{R}^{n_a\times n_g}$ and the vector $\hat{b} =\begin{bmatrix}\underline{u}_{i,p}^\top&\tilde{b}^\top\end{bmatrix}^\top \in\mathbb{R}^{n_a}$, where $n_a$ is the number of active constraints.\\
The objective function in~\eqref{eqn:OCPing} can be expressed as
\begin{align}
    f(g) &= g^\top H g + 2h^\top g + c, 
\end{align}
where
\begin{align}
    H &= U_f^\top \mathcal{R} U_f + Y_f^\top \mathcal{Q} Y_f + \Lambda + \lambda_\sigma Y_p^\top Y_p  \\
    h &= -Y_f^\top  \mathcal{Q} \underline{r}_{i+p,f} - \lambda_\sigma Y_p^\top\underline{y}_{i,p}  \\
    c &= \underline{r}_{i+p,f}^\top  \mathcal{Q} \underline{r}_{i+p,f} + \lambda_\sigma \underline{y}_{i,p}^\top \underline{y}_{i,p}, 
\end{align}
with $ \mathcal{Q} = I \otimes Q\in\mathbb{R}^{n_y\times f}$ and $ \mathcal{R} = I \otimes R\in\mathbb{R}^{n_u\times f}$. Introducing the Lagrange multiplier $\mu\in\mathbb{R}^{n_a}$ yields the Lagrangian
\begin{align}
    \mathcal{L}(g,\mu) &= g^\top H g + 2h^\top g + c + \mu^\top(\hat{A}g-\hat{b}). 
\end{align}
The stationarity and primal feasibility conditions are given by
\begin{align}
\label{eqn:KKTconditions1}
    2Hg^\ast + 2h + \hat{A}^\top\mu^\ast &= 0 \\
    \label{eqn:KKTconditions2}
    \hat{A}g^\ast - \hat{b} &= 0 .
\end{align}
\begin{assum}
\label{ass:LICQ}
    The matrix $\hat{A}\in\mathbb{R}^{n_a\times n_g}$ has full row rank, i.e., $\operatorname{rank}(\hat{A})=n_a$.
\end{assum}
Assumption~\ref{ass:LICQ} is known as the \gls{licq} \cite{bertsekas1999nonlinear} and guarantees that the inverse $(\hat{A}H^{-1}\hat{A}^\top)^{-1}$ exists. Eliminating the Lagrange multiplier $\mu^\ast$ in \eqref{eqn:KKTconditions1}-\eqref{eqn:KKTconditions2} yields the explicit solution
\begin{align}
\label{eqn:minimizerg}
    g^\ast &= -H^{-1}h +H^{-1}\hat{A}^\top
    (\hat{A}H^{-1}\hat{A}^\top)^{-1}
    (\hat{b}+\hat{A}H^{-1}h).
\end{align}
We obtain the following result for standard \gls{deepc} with $\ell_2$-norm regularization and for the proposed weight-\gls{dpc} algorithm.
\begin{prop}
\label{prop:deepc}
The minimizer in~\eqref{eqn:minimizerg} is a piecewise affine function of the past trajectory $\tau_{i,p}$ and is affine on each active set. The past inputs $\underline{u}_{i,p}$ influence the minimizer through $\hat{b}$, while the past outputs $\underline{y}_{i,p}$ influence it through $h$. Then:
\begin{itemize}
    \item For $\lambda_g>0$ and $\lambda_w>0$ the weight-\gls{dpc} formulation is obtained, where the weighting matrix $\Lambda$ is diagonal with non-uniform entries. Consequently, the regularization term induces a non-uniform penalty across the columns of the data matrix, thereby favoring data associated with operating regimes that are deemed more relevant according to the weights defined in~\eqref{eqn:weights}.
    \item For $\lambda_g>0$ and $\lambda_w=0$, the standard \gls{deepc} formulation with $\ell_2$-norm regularization is obtained. The weighting matrix satisfies $\Lambda=\lambda_g I$. Consequently, the regularization term penalizes all columns of the data matrix uniformly and does not encode any preference.
\end{itemize} 
\end{prop}
\begin{IEEEproof}
    Take the hard-initialization limit $\lambda_\sigma\to\infty$, so the past output becomes a hard constraint, recovering $Y_p g= \underline{y}_{i,p}$. Let the weight matrices $\mathcal{Q}$ and $\mathcal{R}$ be negligible relative to $\Lambda$. Furthermore, suppose that there are no active constraints. Then, the minimizer of $\|g\|_\Lambda^2$ subject to $Z_P g=\tau_{i,p}$, with $Z_P:=[U_p^\top\ Y_p^\top]^\top$, is the \emph{weighted} minimum-norm solution
    \begin{align}
    \label{eqn:limitgsolution}
        g^\ast &= \Lambda^{-1}Z_P^\top(Z_P\Lambda^{-1}Z_P^\top)^{-1}\tau_{i,p},
    \end{align}
    which assigns different weights to each column of the data matrix depending on their distance to the current operation point according to~\eqref{eqn:weight_matrix}. For standard \gls{deepc} with $\lambda_w=0$, the minimum norm solution simplifies to
    \begin{align}
        g^\ast &= Z_P^\top(Z_PZ_P^\top)^{-1}\tau_{i,p},
    \end{align}
    where the columns are assigned with uniform weights.
\end{IEEEproof}
Similar to the observations of \cite{11312489}, the limitation of standard \gls{deepc} is not the choice of norm itself but rather the uniform weighting applied to all columns through the regularization term.

\section{Numerical Example}
\label{sec:numericalEx}
The proposed weight-\gls{dpc} framework is evaluated on a nonlinear coupled-tank system \cite{Liquidleveltracking}, whose continuous-time dynamics are given by
\begin{equation}
\begin{aligned}
\label{eqn:WaterTanksys}
    \dot{x}_1(t) &= -0.904\sqrt{x_1(t)} + 0.258 u(t)\\
    \dot{x}_2(t) &= 0.904\sqrt{x_1(t)} - 0.508\sqrt{x_2(t)}\\
    y(t) &= x_2(t).
\end{aligned}
\end{equation}
The state vector is defined as $x(t) = \begin{bmatrix}x_1(t)&x_2(t)\end{bmatrix}^\top$, where $x_1$ and $x_2$ denote the liquid levels in the upper and lower tanks, respectively.
\begin{figure}
    \centering
    \includegraphics[width=\linewidth]{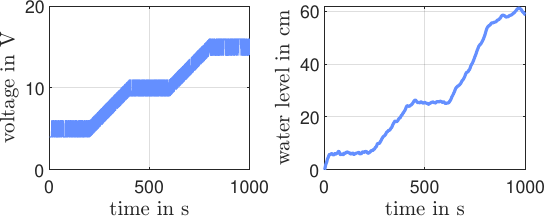}
    \caption{Collected input-output data of the 2-tank system}
    \label{fig:inputoutputdata}
\end{figure}
The continuous-time dynamics are simulated using a fixed integration step of $0.01$~s, while the controller operates with a sampling period of $1$~s. To construct the Hankel matrices, a \gls{pe} input sequence consisting of $1000$ samples is generated. The input signal, shown in the left plot in Figure~\ref{fig:inputoutputdata}, is obtained by superimposing a pseudo-binary random sequence on piecewise-constant operating levels of $5$, $10$, and $15$~V, which are connected through linear transitions. This design ensures sufficient excitation while spanning multiple operating regions of the nonlinear system, making it well suited for analyzing how the proposed weighted regularization influences data selection. The control objective is to minimize the cost function defined in~\eqref{eqn:Objective}, with the performance matrices $Q=1$ and $R=10^{-5}$. The output reference trajectory is shown in gray in Figure~\ref{fig:Output_stagecost}.
\begin{figure}
    \centering
    \includegraphics[width=\linewidth]{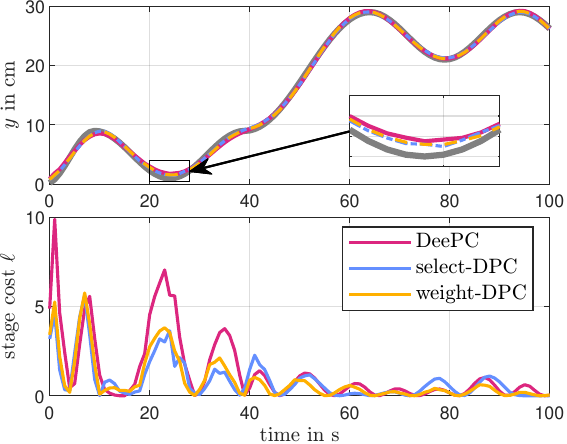}
    \caption{Output and stage cost for \gls{deepc}, select-\gls{dpc}, and weight-\gls{dpc}}
    \label{fig:Output_stagecost}
\end{figure}
The corresponding closed-loop output trajectories obtained with \gls{deepc}, select-\gls{dpc}, and weight-\gls{dpc} are shown in red, blue, and yellow, respectively.\\
For each controller, the hyperparameters are tuned through a grid search. The parameter values yielding the lowest cumulative stage cost, $\sum_k \ell(u_k,y_k)$, are reported in Table~\ref{tab:parameters}. For select-\gls{dpc}, $n_g$ reduces to the number of columns retained in the data matrix after the selection procedure.
\begin{table}
    \caption{Optimal controller parameters}
    \centering
    \begin{tabular}{c|c|c|c|c|c|c}
         & $p$ & $f$ & $\lambda_g$ & $\lambda_\sigma$ & $\lambda_w$ & $n_g$\\
         \hline
         \gls{deepc}& $5$ & $19$ & $60$ & $10^4$ & - & $1000$\\
         select-\gls{dpc}& $5$ & $17$ & $40$ & $10^5$ & - &400 \\
         weight-\gls{dpc} & $5$ & $15$ & $0.9$ & $10^4$ & $1300$ & $1000$
    \end{tabular}
    \label{tab:parameters}
\end{table}
The resulting cumulative stage costs at the optimal parameter settings are $132.7$, $85.8$, and $84.9$ for \gls{deepc}, select-\gls{dpc}, and weight-\gls{dpc}, respectively. These results demonstrate a performance improvement of both data-selection approaches over standard \gls{deepc}. Furthermore, weight-\gls{dpc} achieves the lowest overall cost, indicating that weighting the data according to operating point relevance can be at least as effective as explicit data selection while preserving the full data matrix. For the first $40$~s, the system is operated with a water level in the region of the first operation point of the data collection in Figure~\ref{fig:inputoutputdata}. The stage cost of select-\gls{dpc} and weight-\gls{dpc} are lower than for standard \gls{deepc}. After 40~s, the reference output transitions into the region of the second operating point. In this regime, weight-\gls{dpc} achieves the lowest stage cost. As the system exhibits stronger nonlinear behavior at low water levels, all three controllers attain a lower stage cost at the higher operating point.\\
Figure~\ref{fig:distance_heatmap} shows the absolute value of each entry of the column selector $g$ over time for the simulation presented in Figure~\ref{fig:Output_stagecost} in blue.
\begin{figure}
\centering
\begin{minipage}{0.74\linewidth}
    \includegraphics[width=\linewidth]{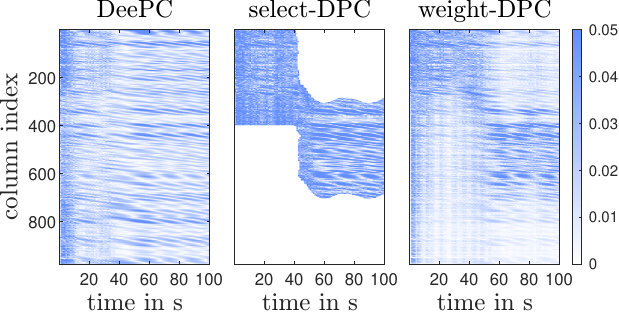}
\end{minipage}\hfill
\begin{minipage}{0.23\linewidth}
    \includegraphics[width=\linewidth]{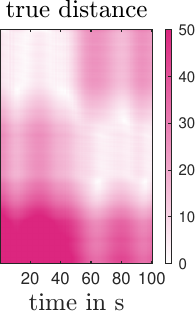}
\end{minipage}
\caption{Absolute value of the entries of $g$ (blue). True distance between the state $x_{i+p}$ and the state $x_{v+p}$ associated with each data column (red).}
\label{fig:distance_heatmap}
\end{figure}
Darker shades indicate larger values of $|g_v|$, implying that the $v^{th}$ column of the data matrix contributes more to the prediction, whereas lighter shades indicate a smaller contribution.\\
For standard \gls{deepc}, the heatmap reveals that all columns are utilized to a similar extent throughout the simulation. In contrast, select-\gls{dpc} initially relies solely on the first $400$ columns, since the first $400$ samples in Figure~\ref{fig:inputoutputdata} correspond to operating conditions that are closest to those encountered during the reference tracking task, whereas columns with indices exceeding $400$ originate from regions with higher water levels and larger input voltages. Consequently, these columns are excluded by the selection procedure. After approximately $40$~s, the operating region shifts, causing select-\gls{dpc} to utilize columns with indices roughly between $300$ and $700$.\\
The heatmap for weight-\gls{dpc} shows a markedly different behavior. Rather than discarding columns, all data columns remain available for prediction. However, columns associated with operating points that are close to the current system state receive substantially larger weights and therefore contribute more strongly to the prediction.\\
To further examine the influence of the operating point, the true distance between the state at the time of prediction of weight-\gls{dpc} and the state $x_{v+p}$ associated with each column of the data matrix is shown in red in Figure~\ref{fig:distance_heatmap}. Lighter colors indicate smaller distances and thus columns that are expected to be more relevant for prediction. The distance patterns obtained for standard \gls{deepc} and select-\gls{dpc} are nearly identical to that in Figure~\ref{fig:distance_heatmap}. A clear correspondence between state proximity in red and data utilization in blue can be observed for both select-\gls{dpc} and weight-\gls{dpc}, which explains the superior performance compared to standard \gls{deepc}.\\
While select-\gls{dpc} performs a binary selection by either retaining or discarding a column and subsequently treats all retained columns equally, weight-\gls{dpc} establishes a continuous preference over the available data by assigning weights according to the proximity of each column to the current operating point. As a result, the method exploits the entire dataset while still emphasizing the most relevant data.
\section{Conclusion}
\label{sec:Conclusion}
We propose a novel \gls{dpc} framework for nonlinear systems by incorporating a weighted norm regularization term. We explicitly derive how the weighted regularization influences data usage and compare the proposed approach with both standard \gls{deepc} and select-\gls{dpc}. Numerical simulations demonstrate that the proposed method outperforms standard \gls{deepc} and achieves similar performance to select-\gls{dpc}. In contrast to select-\gls{dpc}, however, the proposed approach preserves the full data matrix and therefore guarantees that the required rank condition remains satisfied at all times.\\
Future research could explore alternative metrics for measuring the proximity of operating points based on recorded input-output data, as well as projection-based weighted regularization.

\bibliographystyle{IEEEtran}
\bibliography{literature}
\end{document}